\documentclass{article}

\usepackage{amsmath,amssymb,amsthm}
\usepackage[utf8]{inputenc}
\usepackage{hyperref}

\newtheorem{theorem}{Theorem}
\newtheorem{lemma}[theorem]{Lemma}
\newtheorem{corollary}[theorem]{Corollary}
\theoremstyle{plain}
\newtheorem{claim}[theorem]{Claim}

\usepackage[subpreambles=true,mode=buildnew]{standalone}

\usepackage[
    natbib = true,
    backend=biber,
    isbn=false,
    url=false,
    doi=false,
    maxbibnames=5,
    minbibnames=5,
    eprint=false,
    style=numeric,
    sorting=nyt,
    sortcites = true
]{biblatex}
\bibliography{cfan}

\DeclareNameAlias{author}{first-last}

\title{On Visibility Graphs of Convex Fans and Terrains}

\usepackage{authblk}

\author{André C. Silva\thanks{andre.silva@ic.unicamp.br}}

\usepackage{lineno}

\begin{document}

\maketitle

\begin{abstract}
	For two points in the closure of a simple polygon $P$, we say that they \textit{see} each other in $P$ if the line segment uniting them does not intersect the exterior of $P$. The \textit{visibility graph} of $P$ is the graph whose vertex set is the vertex set of $P$ and two vertices are joined by an edge if they see each other in $P$.
	The characterization of visibility graphs has been an open problem for decades, and a significant effort has been made to characterize visibility graphs of restricted polygon classes. Among them is the \textit{convex fan}: a simple polygon with a convex vertex (i.e. whose internal angle in less than 180 degrees) that sees every other point in the closure of the polygon (often called \textit{kernel vertex}). We show that visibility graphs of convex fans are equivalent to visibility graphs of terrains with the addition of a \textit{universal vertex}, that is, a vertex that is adjacent to every other vertex.
\end{abstract}

\section{Introduction}

We represent a polygon or polygonal chain $P$ as a as a finite cyclic sequence of points, its \textit{vertices}, in the plane. For any two consecutive vertices $v$ and $w$ of $P$, the segment $vw$ uniting $v$ and $w$ is one of its \textit{edges}. The union of all its edges is the \textit{boundary} of $P$. A polygon is \textit{simple} if its edges only intersect at their endpoints. The boundary of any simple polygon $P$ will enclose a path-connected bounded set of points in the plane. These points and its boundary form the \textit{closure} of $P$. Polygons here are assumed to be in \textit{general position}, which means no three vertices of the polygon are collinear.

Visibility graphs of polygons is a structure that naturally arises in many problems regarding visibility in polygons, such as the art gallery problem~\cite{agta} and euclidean shortest path problem~\cite{esp}. A complete characterization remains an open problem to this day. Naturally, restricting the problem to particular graph or polygon classes is a common strategy in dealing with it.

A class that got a lot of attention in the literature is the convex fans. The characterization of visibility graphs of convex fans has been the source of a lot of confusion in the visibility graph literature. It has been claimed a few times~\cite{hdpwa,vgsp,vgom}, however, each time, either the claim was erroneous or the results were never published. Since we can decompose a polygon into a sequence of convex fans, a characterization of their visibility graphs may help advance the general problem.

A convex fan is \textit{orthogonal} if all its edges are either vertical or horizontal. Orthogonal convex fans are also referred as \textit{staircase polygons}. It is clear that any vertex, besides the kernel vertex, with an internal angle of 90 degrees is only visible to the kernel vertex and the two other vertices that share an edge with it. Thus, it makes sense to consider only the \textit{core} of a staircase polygon, that is, all the vertices with an internal angle different from 90 degrees. 

Abello, Egecioglu and Kumar~\cite{vgsp} showed that the visibility graph of a staircase polygon induced by its core is \textit{persistent}, that is,  they satisfy the following properties\footnote{These properties were originally presented by Abello et al. as properties on the adjacency matrix of the graph. Evans and Saeedi~\cite{octvg} translated them into properties on the graph and renamed them. We use their version. }: (X-property) For any four vertices $a<b<c<d$, if $ac$ and $bd$ are edges in the graph, then $ad$ is also an edge; (Bar property) For any edge $ac$ with $c < a+1$, there exists a vertex $b$ adjacent to both $a$ and $c$ such that $a < b < c$. They also provided an algorithm to build a generalized configuration of points from a persistent ordered graph (i.e. whose vertices are totally ordered). Evans and Saeedi~\cite{octvg} simplified their proof, resulting in a faster algorithm.

Colley~\cite{rvgup} shows that the visibility graph of the core of a staircase polygon is equivalent to the visibility graph of a \textit{terrain}: an x-monotone polygonal chain. Note that terrains cannot be a closed curve (and thus, a polygon) which makes visibility a bit different. Two vertices of a terrain \textit{see} each other if and only if no vertex of the terrain is above or on the segment uniting them.

We expand upon the literature and reduce the characterization of visibility graphs of convex fans to visibility graphs of terrains.

\begin{theorem}\label{thm:cfit}
	A graph $G$ is the visibility graph of a convex fan $F$  if and only if for some universal vertex $v$ of $G$, $G-v$ is a visibility graph of a terrain $T$. 
\end{theorem}

One of the consequences of Theorem \ref{thm:cfit} is that the \textit{recognition} and \textit{reconstruction} problems (defined below) of convex fans and terrains are equivalent under polynomial time reductions. The next theorem captures this fact. 

\begin{theorem}\label{thm:rreq}
	The recognition and reconstruction problems for convex fans are equivalent, under polynomial-time reductions, to its respective problems for terrains.
\end{theorem}

The \textit{recognition} problem for visibility graphs asks, for a given input graph $G$, to decide whether $G$ is the visibility graph of some polygon (or a terrain). Similarly, the \textit{reconstruction} problem concerns algorithms that outputs a polygon (terrain) whose visibility graph is isomorphic to a visibility graph given as an input. The only known result about the computational complexity of these problems are by Everett~\cite{vgr}, who showed visibility graphs reconstruction is in PSPACE. So it is natural to consider restricted versions of them, by restricting the input graph.

Before we proceed with the proof of Theorem~\ref{thm:cfit}, in Section~\ref{sec:pfthm}, we prove a few properties of central projections (defined below) in the next section. In Section~\ref{sec:rrp}, we prove Theorem~\ref{thm:rreq}.

\section{Brief discussion on central projections}
\label{sec:cp}

\begin{figure}
	\centering
	\label{fig:cpjbp}
	\includestandalone[scale=.8]{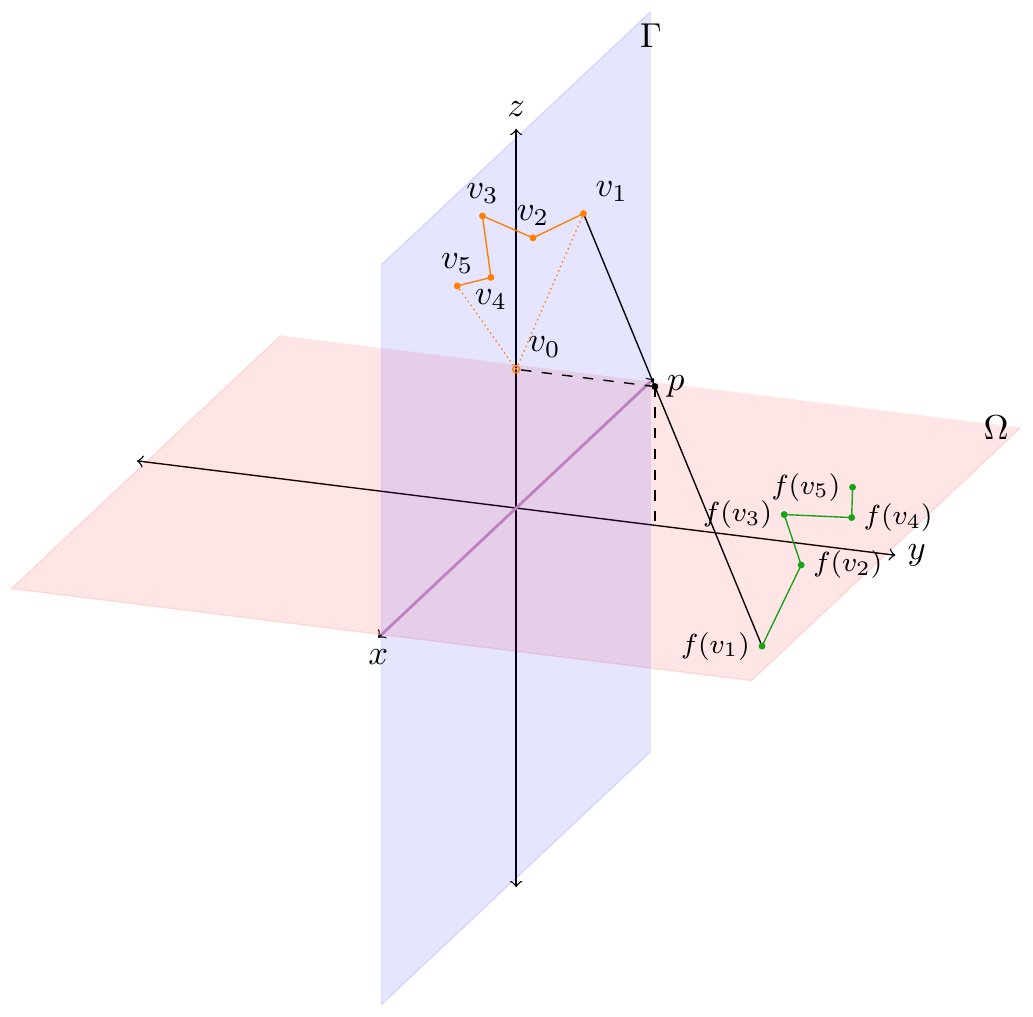}
	\caption{A central projection between $\Gamma$ and $\Omega$.}
\end{figure}

In this section we highlight some properties of central projections (defined below) --- a concept borrowed from projective geometry. The results and concepts here are by no means new nor groundbreaking, but they serve as a foothold for the proof of Theorem \ref{thm:cfit}.

A \textit{central projection} $f$ between planes $\Gamma$ and $\Omega$ in $\mathbb{R}^3$ through a point $p$ not in $\Gamma$ nor $\Omega$ is a function that maps a point $a \in \Gamma$ to the intersection between the line $pa$ and $\Omega$ (see Figure~\ref{fig:cpjbp}). Naturally, for some points $v$ of $\Gamma$, the line $vp$ may be parallel to $\Omega$. In this case, we name $v$ a \textit{vanishing point} of $\Gamma$. There exists only one plane parallel to $\Omega$ that contains $p$. Any line parallel to $\Omega$ containing $p$ is in this plane. Thus, its intersection with $\Gamma$ contains all the vanishing points of $\Gamma$. This implies that all of them are collinear, forming a \textit{vanishing line}. Note that $\Omega$ also has vanishing points, that is, for some point $u$ of $\Omega$, $pu$ is parallel to $\Gamma$, and thus, also contains a vanishing line. 

If $l$ is a line in $\Gamma$, then the images of all points in $l$ under $f$ are in a line $m$ of $\Omega$. To see this, let $\Lambda$ be the unique plane containing $p$ and $l$. For any point $x$ of $l$, the line $px$ is in $\Lambda$. Thus, $f(x)$ lies on the intersection of $\Lambda$ and $\Omega$, which is a line. %

We note that $f$ is an homeomorphism between $\Gamma - l$ and $\Omega -m$ where $l$ and $m$ are vanishing lines of $\Gamma$ and $\Omega$, respectively. Thus $f$ has an continuous inverse $f^{-1}$ that is also a central projection. The following facts are fairly straightforward consequences of this.

\begin{corollary}\label{cor:hpthp}
	Let $f$ be a central projection between the planes $\Gamma$ and $\Omega$ with $l$ and $m$ as its respective vanishing lines. Then, an open half-plane of $\Omega$ bounded by $l$ is mapped onto an open half-plane bounded by $m$. 
\end{corollary}

\begin{corollary}\label{cor:pib}
	Let $a$, $b$ and $c$ be collinear points of a plane $\Gamma$ and let $f$ be a central projection from $\Gamma$ to another plane. If $b$ is between $a$ and $c$ in $\Gamma$, but no vanishing point is, then $f(b)$ is also between $f(a)$ and $f(c)$, and no vanishing point is. 
\end{corollary}

\begin{figure}
	\centering
	\label{fig:pib}
	\includestandalone[scale=1.2]{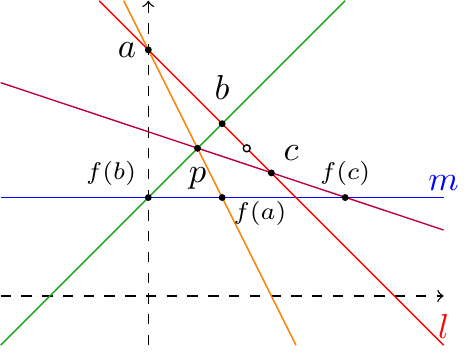}
	\caption{While $b$ is in between $a$ and $c$, $f(b)$ is not between $f(a)$ and $f(b)$.}
\end{figure}

Corollary \ref{cor:pib} may not be true if some vanishing point is between $a$ and $b$. See Figure~\ref{fig:pib} for a counter-example. 

\begin{lemma}\label{lem:mvppl}
	Let $f$ be a central projection from $\Gamma$ to $\Omega$ and let $n$ be the vanishing line of $\Gamma$. For lines $m$ and $l$ of $\Gamma$, their images are in a pair of parallel lines only if either both are parallel to $n$ or the meet at a vanishing point. 
\end{lemma}
\begin{proof}
	Recall that collinear points in $\Gamma$ are mapped into collinear points in $\Omega$. We shall name $\overline{f(l)}$ and $\overline{f(m)}$ the lines that contains the images of the points in $l$ and $m$, respectively. If $l$ meets $m$ in a non-vanishing point, then $\overline{f(l)}$ and $\overline{f(m)}$ meet at a non-vanishing point, as a consequence of Corollary \ref{cor:pib}.
	
	Suppose $l$ and $m$ have a vanishing point $v$ in common. Consider the unique plane $\Lambda$ that contains $l$ and $p$. Any line from $p$ and a point of $l$ must be in $\Lambda$, thus the intersection of $\Lambda$ and $\Omega$ is exactly $\overline{f(l)}$. Similarly for the plane $\Sigma$ containing both $m$ and $p$. As $\Lambda$ and $\Sigma$ both contain $p$ and $v$, the line $pv$ is the exact intersection of $\Lambda$ and $\Sigma$. Any point belonging to both $\overline{f(l)}$ and $\overline{f(m)}$ has to be on this line, however this line is parallel to $\Omega$. Thus, we conclude that $\overline{f(l)}$ and $\overline{f(m)}$ cannot meet in $\Omega$ and therefore are parallel. 

	Now suppose $l$ and $m$ are parallel to $n$. Then $l$ is entirely contained in a half-plane of $\Gamma$ bounded by $n$. Similarly for $m$. By Corollary~\ref{cor:hpthp}, $\overline{f(l)}$ is also contained in a half-plane of $\Omega$ bound by the vanishing line of $\Omega$. This means that $\overline{f(l)}$ is parallel to it. Similarly for $\overline{f(m)}$, so they are parallel to each other. 

\end{proof}

\section{Proof of Theorem \ref{thm:cfit}}
\label{sec:pfthm}

\begin{proof}
Let $F = (v_0,...,v_{n-1})$ be a convex fan in general position where $v_0$ is a convex kernel vertex. We show that there exists a central projection $f$ that maps the polygonal chain $(v_1, \allowbreak v_2, \allowbreak ..., \allowbreak v_{n-1})$ to a terrain $T$ while maintaining visibility between its vertices.

We assume that $\mathbb{R}^3$ has the standard basis. Let $\Gamma$ and $\Omega$ be the $x=0$ and $z=0$ planes. Let $p$ be any point with positive $z$ and $x$ coordinates\footnote{Save to show that $T$ is $x$-monotone, the proof does not rely on these particular planes or point. If $\Gamma$ and $\Omega$ are two meeting planes and $p$ is a point not in them, then $T$ would be monotone to some line.} Let $f$ be the central projection from $\Gamma$ to $\Omega$ through $p$. See Figure~~\ref{fig:cpjbp} for a visual aid. 

We assume $v_0$ is the orthogonal projection of $p$ onto $\Gamma$. This makes $v_0$ a vanishing point of $\Gamma$ with relation to $f$. Additionally, we make sure that the rest of $F$ is drawn above the vanishing line of $\Gamma$. This is possible because $v_0$ is a convex vertex.

\begin{claim}\label{clm:ftot}
	 The polygonal chain $T=(f(v_1),f(v_2),...,f(v_{n-1}))$ is a terrain.
\end{claim}
\begin{proof}
	Since, with the exception of $v_0$, the entirety of $F$ is contained in one of the open half-planes of $\Gamma$ bound by its vanishing line, by Corollary \ref{cor:hpthp}, all the edges of $F$ that are not incident with $v_0$ do not contain a vanishing point. 

	If any two edges $f(v_i)f(v_{i+1})$ and $f(v_j)f(v_{j+1})$ of $T$ share a point between their ends, then, by Corollary \ref{cor:pib}, so does $v_iv_{i+1}$ and $v_jv_{j+1}$. Thus $T$ is a simple polygonal line. It remains to show that $T$ is $x$-monotone.

	Due to the general position hypothesis, for any point $a$ of the polygonal chain $(v_1,...,v_n)$, the line $v_0a$ intersects it solely at $a$. 

	By Lemma \ref{lem:mvppl}, the lines that meet at $v_0$ in $\Gamma$ are mapped into lines that are pairwise parallel in $\Omega$. To conclude that $T$ is $x$-monotone, it suffices to show that one of them is parallel to the $y$-axis in $\Omega$. 
	
	Consider the unique plane that contains $v_0$ and $p$ and is also orthogonal to both $\Gamma$ and $\Omega$. The intersection of this plane and $\Omega$ is a line orthogonal to $\Gamma$, and thus parallel to the $y$-axis, that contains the image of some line that goes through $v_0$ in $\Gamma$. 

\end{proof}

Recall that $f^{-1}$ is also a central projection. A consequence of the previous claim is that $F = (v_0,\allowbreak f^{-1}(v_1), \allowbreak f^{-1}(v_2), \allowbreak ..., \allowbreak f^{-1}(v_{n-1}))$, that is, we may use $f^{-1}$ to obtain a convex fan from a terrain. Therefore, to complete the proof of Theorem \ref{thm:cfit} it suffices to show the following:

\begin{claim}\label{clm:ftvis}
	For $i \in \{1,...,n\}$, the vertices $v_i$ and $v_j$ see each other in $F$ if and only if $f(v_i)$ and $f(v_j)$ see each other in $T$. 
\end{claim}
\begin{proof}
	A consequence of Corollary \ref{cor:hpthp}, is that no vanishing point is in the segment $v_iv_j$, similarly for $f(v_i)f(v_j)$.
	
	Let $v_i$ and $v_j$ be a pair of vertices that see each other in $F$. If $v_iv_j$ is an edge of $F$, we are done. Thus, suppose not. 
	
	Due to the general position hypothesis, for any point $b$ in the segment $v_iv_j$, $b$ is not a point on the boundary of $F$. Let $c$ be any point in the segment $bv_0$.

	Note that the ray $v_0b$ intercepts the boundary of $F$ at a point $a$. Thus $c$ is in the segment $v_0a$ and $f(c)$ is above $T$. 
	
	By Corollary~\ref{cor:pib}, $f(b)$ is also not a point on the boundary of $T$. Pick a point $d$ in the segment $cv_0$. Thus, $c$ is between $b$ and $d$ and so $f(c)$ is between $f(b)$ and $f(d)$. If $f(c)$ is a point on the boundary $T$ then so is $c$ a point on the boundary of $F$. This either contradicts the fact that $F$ is in general position (i.e. $c$ is a vertex) or the fact that $v_0$ sees $b$.

	This means that no point on the boundary of $T$ is above or in the segment $f(v_i)f(v_j)$, and we conclude that $f(v_i)$ and $f(v_j)$ see each other.

	Since $f^{-1}$ is also a central projection. A similar argument, shows that if $f(v_i)$ and $f(v_j)$ see each other, then so are $f^{-1}(f(v_i))=v_i$ and $f^{-1}(f(v_j))=v_j$. 
\end{proof}
\renewcommand{\qedsymbol}{$\blacksquare$}
\end{proof} 

\section{Recognition and reconstruction problems}
\label{sec:rrp}

In this section, we prove Theorem \ref{thm:rreq}. 

\begin{proof}

Let $f$ be the central projection described in Section~\ref{sec:pfthm}. We first reduce the reconstruction and recognition problems of convex fans to the similar problems for terrains.

Let $G$ be an arbitrary graph. If $G$ has no universal vertex then $G$ is not a visibility graph of a convex fan and we stop. For each universal vertex $v_i$ of $G$, let $G_i:=G-v_i$. For supply each $G_i$ as an input for an algorithm that recognizes terrains. If any of them is accepted, then we accept $G$. Theorem \ref{thm:cfit} guarantees $G$ will be accepted if and only if one of the $G_i$ will.

For the reconstruction problem, supply each $G_i$ to an algorithm that reconstructs terrains. Theorem~\ref{thm:cfit} guarantees that at least one of them will output a terrain whose visibility graph is isomorphic to the input if and only if $G$ is a visibility graph of a convex fan. We now use $f^{-1}$ a in Claim~\ref{clm:ftot} to obtain a convex fan whose visibility graph is isomorphic to $G$. 

We now make the reductions in the other direction. Let $G$ be an arbitrary graph and let $G'$ be obtained from $G$ by adding a universal vertex $v$. For the reconstruction problem, we supply $G'$ as an input for an algorithm that recognizes visibility graphs of convex fans and accept $G$ if $G'$ is also. Again Theorem~\ref{thm:cfit} shows that $G'$ will be accepted if and only if $G$ is accepted. 

For the reconstruction problem. We supply $G'$ to an algorithm that reconstructs convex fans. Theorem~\ref{thm:cfit} will guarantee that $G'$ will output a convex fan. We then use $f$ to obtain a terrain as in Claim~\ref{clm:ftot}. 

\end{proof}

\printbibliography

\end{document}